\documentclass{article}
\usepackage{graphicx} % 
\usepackage{multicol}
\usepackage{xcolor}
\usepackage{verbatim}
\usepackage{amsmath,amssymb,amsthm,graphicx}
\usepackage{arxiv}
\usepackage{hyperref}
\usepackage{amsthm}
\usepackage{multirow}
\usepackage{hyperref}
\usepackage{comment}

\newtheoremstyle{thm}% name
{9pt}%      Space above, empty = `usual value'
{9pt}%      Space below
{\itshape}% Body font
{}%         Indent amount (empty = no indent, \parindent = para indent)
{\bfseries}% Thm head font
{.}%        Punctuation after thm head
{ }% Space after thm head: \newline = linebreak
{}%         Thm head spec
\theoremstyle{thm}
\newcommand{\sgn}{\operatorname{sgn}}
\newtheorem{theorem}{Theorem}[section]
\newtheorem{lemma}[theorem]{Lemma}
\newtheorem{corollary}[theorem]{Corollary}

\usepackage{algorithm}
\usepackage[noend]{algpseudocode}
\makeatletter
\usepackage{arydshln}

\title{Testing independence in the presence of missing data: high-dimensional case
}

\author{ \href{https://orcid.org/0000-0001-5071-8350}{\includegraphics[scale=0.06]{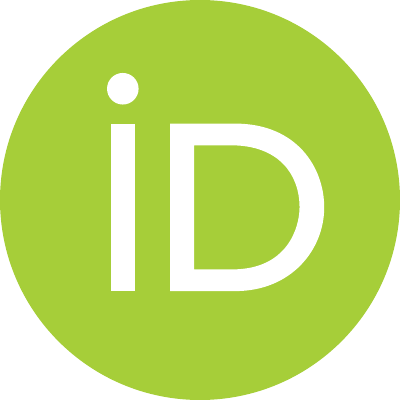}\hspace{1mm}Marija Cupari\' c} \\
	University of Belgrade\\
        Faculty of Mathematics\\
	\texttt{marija.cuparic@matf.bg.ac.rs}
 \And	 
 \href{https://orcid.org/0000-0001-8243-9794}{\includegraphics[scale=0.06]{orcid.pdf}\hspace{1mm}Bojana Milo\v sevi\' c} \\
	University of Belgrade\\
        Faculty of Mathematics\\
	\texttt{bojana@matf.bg.ac.rs}
 \And	 
      \href{https://orcid.org/0009-0001-7275-6654}{\includegraphics[scale=0.06]{orcid.pdf}\hspace{1mm}Jelena Radojevi\' c} \\
	University of Belgrade\\ 
        Faculty of Civil Engineering\\ 
        Faculty of Mathematics \\
\texttt{jradojevic@grf.bg.ac.rs}
}

\begin{document}

\maketitle 
\begin{abstract} 
In this paper, we consider the problem of testing independence in high-dimensional settings with missing data. Building upon a recently proposed Kendall-based statistic, we introduce two new modifications specifically designed to accommodate incomplete observations. The proposed methods are studied from both theoretical and empirical perspectives. A comprehensive simulation study illustrates the robustness and applicability of the new approaches. The findings contribute to the development of nonparametric methods for analyzing high-dimensional and incomplete data structures.

\keywords{high-dimensional data, missing values, independence testing, Kendall coefficient, M(C)AR}
\end{abstract}

\section{Introduction}
\sloppy

In many modern scientific fields, such as genomics, biomedical research, finance, environmental studies, and machine learning, a typical property of underlying datasets is that the number of variables $d$ is often comparable to, or even much larger than, the sample size $n$. This high-dimensional setting requires statistical procedures whose theoretical properties remain valid when $d,n \to\infty$. Over the past decade, numerous model specification tests in high dimensions have been developed to address these challenges (see e.g. \cite{xu2016adaptive}, \cite{zhang2021simple}, \cite{zhang2024projective}, \cite{sang2024test}, \cite{marozzi2020interpoint} and reference therein).

A second and equally common challenge in modern data analysis is the presence of missing data. For readers interested in practical approaches to handling missing data, we recommend the monograph by \cite{enders2022applied}. Missingness may arise due to human error, measurement limitations, technical problems, or cost constraints, and it can significantly impact the accuracy and interpretability of statistical procedures. For example, in the context of survey data, a comprehensive overview of this issue is provided by \cite{mirzaei2022missing}. 

Let  $\mathbf{X}$ be a $d$-dimensional random vector, and 
$\mathbf{R} = (R^{(1)}, \dots, R^{(d)})$ the corresponding \emph{response indicator vector}, where $R_i = 1$ if $X_i$ is observed and $0$ otherwise. The conditional distribution of $R$ given $X$ defines the \emph{missingness mechanism}. Following \cite{little2019statistical}, data are classified as:
\begin{itemize}
    \item \textbf{Missing Completely at Random (MCAR)} if 
    $$
    P(\mathbf{R} = \mathbf{r} \mid \mathbf{X}_{\mathrm{obs}}, \mathbf{X}_{\mathrm{mis}}) = P(\mathbf{R} = \mathbf{r}),
    $$
    \item \textbf{Missing at Random (MAR)} if
    $$
    P(\mathbf{R} = \mathbf{r} \mid \mathbf{X}_{\mathrm{obs}}, \mathbf{X}_{\mathrm{mis}}) = P(\mathbf{R} = \mathbf{r}  \mid \mathbf{X}_{\mathrm{obs}}),
    $$
    \item \textbf{Missing Not at Random (MNAR)} otherwise.
\end{itemize}
Depending on the type of missingness, different approaches for handling this problem are available in the literature. In practice, complete-case analysis - removing observations with any missing entries - is one common approach. Although, under MCAR, this approach usually yields unbiased estimates and possesses many desirable statistical properties, it results in having no observations when $d$ is large.

Another widely used approach is imputation. Although numerous imputation methods have been proposed (see, e.g., \cite{Buuren2012, stekhoven2012missforest, golchian2025missing}), simple techniques such as mean or median imputation are still frequently applied in practice. However, such strategies may distort the underlying distributions and introduce bias in the estimation of variances, correlations, and regression coefficients \cite{little2019statistical}. Moreover, the missingness mechanism is often unknown \cite{xie2017dissecting}, which makes naive imputation and complete-case analysis unreliable in many situations.

For testing the MCAR assumption, we refer to the classical Little’s test \cite{little1988test}, as well as to more recent developments \cite{aleksic2024novel, bordino2025tests}.

Testing total independence among variables becomes challenging because many dependence measures rely on fully observed pairs of observations. Traditional high-dimensional independence tests, including rank-based statistics and kernel methods, typically assume complete data (see \cite{mao2018testing} and references therein).

To formalize the problem, we consider the following statistical framework for testing total independence among multiple variables.

Let $\boldsymbol{X} = (X_1, \dots, X_d)$ be a jointly and marginally continuous random vector with joint density $f(x_1, \dots, x_d)$ and marginal densities $f_k(x_k)$ for $k=1,\dots,d$. Based on this setup, the null hypothesis of total independence among the components of $\boldsymbol{X}$ is
    \begin{align*}
        H_0: f(x_1, \dots, x_d) = \prod_{k=1}^{d} f_k(x_k),
    \end{align*}
    against the alternative
    \begin{align*}
        H_1: H_0 \text{ does not hold.}
    \end{align*}
    
In the Gaussian case, this is equivalent to pairwise independence, which leads to test statistics based on the sum of squares of all pairwise Pearson correlation coefficients (see \cite{schott2005testing}). In general, this equivalence does not hold. However, statistics of a similar form have been proposed in which Pearson’s correlation coefficient is replaced by other measures of dependence. Moreover, when the Gaussian assumption does not hold other measures are preferable. 
Let us assume that we have a independent and identically distributed (i.i.d.) sample $\mathbf{X}_1,\dots, \mathbf{X}_d$ as $\mathbf{X}$. One of the most popular choices are rank-based dependence measures including the Kendall's $\tau$. In particular, Kendall's $\tau$ for $X_k$ and $X_l$ is defined as
\begin{equation*}
    \tau_{kl} = \frac{2}{n(n-1)}\sum_{i=2}^n \sum_{j=1}^{i-1} \sgn(X_{ik} - X_{jk})\sgn(X_{il}-X_{jl}).
\end{equation*}
In \cite{Leung2018} (see also \cite{mao2018testing}), the corresponding test statistic is proposed
\[
T_{nd} = \sum_{k=2}^{d} \sum_{l=1}^{k-1} \tau_{kl}^2.
\]
The author showed that under $H_0$ $$S_\tau^* = \frac{1}{\sigma_{nd}} \left(T_{nd} - \frac{d(d-1)(2n+5)}{9n(n-1)} \right),$$ where $$\sigma_{np}^2 = \operatorname{Var}(T_{nd}) = \frac{4d(d-1)(n-2)(100n^3+492n^2+731n+279)}{2025n^3(n-1)^3},$$ converges in distribution to a standard normal random variable as both $n$ and $d$ grow.
This establishes the asymptotic validity of the test and justifies using the standard normal quantiles to assess significance in high-dimensional settings.

In this paper, we extend this framework to datasets with missing values under the MCAR assumption. We propose modifications that ensure robust and principled inference. Our approach is valid under high-dimensional asymptotics $d,n \to \infty$. Our method is based on the sum of squares of pairwise Kendall's $\tau$ coefficients, allowing valid testing of total independence without relying on Gaussianity or fully observed data.

The paper is organized as follows. In Section \ref{sec:testStat}, we first consider a complete-case approach for testing independence based on a modified Kendall’s $\tau$ statistic in the presence of MCAR data. Since this approach is not feasible in high-dimensional settings without additional restrictions on the completeness probabilities, we then propose a novel pairwise test that exploits all available observations. For the proposed statistic, we derive the asymptotic null distribution, which is subsequently used to approximate p-values. In Section \ref{sec:empStudy}, we present the results of an extensive empirical study assessing the size and power of the proposed test under various settings. Concluding remarks are given in Section \ref{sec:conclusion}. All proofs are provided in the Appendix.

\section{Test statistics}\label{sec:testStat}
\subsection{Complete-case approach}

In the presence of missing data, a natural strategy is to restrict the analysis to fully observed samples. Let $S_i = \prod_{k=1}^d R_{ik}$ denote the completeness indicator for the $i$-th observation, where $R_{ik} = 1$ if $X_{ik}$ is observed and $0$ otherwise. Following \cite{aleksic2023non}, the Kendall's $\tau$ statistic can then be modified as
\begin{equation}
\tau_{kl}^{cc} = \frac{2}{n(n-1)} \sum_{i=2}^n \sum_{j=1}^{i-1} 
\sgn(X_{ik} - X_{jk}) \, \sgn(X_{il} - X_{jl}) \, S_i S_j, 
\end{equation}
and the overall test statistic for $d$-dimensional data is
\begin{equation}
T_{nd}^{cc} = \sum_{k=2}^{d} \sum_{l=1}^{k-1} (\tau_{kl}^{cc})^2.
\end{equation}

Under the null hypothesis of mutual independence $H_0$, and assuming the missingness mechanism is MCAR with $ES_i = q$, the expectation of the statistic is given by
\begin{equation}
E(T_{nd}^{cc}) = \frac{2 d(d-1) q^2}{n(n-1)} \left( \frac{n-2}{9} q + \frac{1}{2} \right).
\end{equation}
The variance can be derived as
\begin{align*}
\operatorname{Var}&(T_{nd}^{cc}) \!=\! \frac{4 d(d-1) q^2}{2025 n^3 (n-1)^3} \Bigl( 2025 \!+\! 16650 (-2\!+\!n) q \!+\! 675 (111 \!-\! 91 n \!+\! 17 n^2) q^2 \\
&\!+\! 18 (-3456 \!+\! 3694 n \!-\! 1221 n^2 \!+\! 119 n^3) q^3 \!+\! 50 (360 \!-\! 458 n \!+\! 205 n^2 \\ &\!-\! 37 n^3 \!+\! 2 n^4) q^4 \Bigr)
\!-\! \frac{2(d+1)d(d-1)(d-2) q^2 (q-1)}{27 n^3 (n-1)^3}\\ &\cdot\! \Bigl( 27 \!+\! (\!-\!161\!+\!94 n) q \!+\! 8 (26-25 n+6 n^2) q^2 \!+\! 2(-40 \!+\! 50 n \!-\! 21 n^2 \!+\! 3 n^3) q^3 \Bigr).
\end{align*}

These results motivate an empirical investigation of the finite-sample behavior of the standardized test statistic under the MCAR mechanism. Our simulations indicate that, for fixed $q$, the standardized statistic is well approximated by a normal distribution under the null hypothesis.
We do not investigate the limiting properties of the statistic due to the aforementioned limitations of the complete-case approach, which prevent a meaningful assessment in high-dimensional settings. Despite these favorable theoretical and empirical properties, the complete-case approach becomes infeasible as the dimension $d$ increases with fixed column-wise missingness, since the number of fully observed observations rapidly decreases. Although feasibility can be preserved when the total proportion of missingness is fixed, such an assumption is generally unrealistic in high-dimensional scenarios, thereby limiting the practical applicability of the complete-case methodology.

\subsection{Pairwise approach}

Motivated by the limitations of the complete-case approach in high-dimen\-sional settings, we propose an alternative strategy that avoids discarding partially observed samples. The key idea is to replace the global completeness indicators $S_i$ and $S_j$ with appropriate cell-wise missingness indicators, thereby exploiting all available pairwise information.

Specifically, we construct a modified Kendall’s tau statistic in which each pairwise comparison is included only if the corresponding observations are available for both variables, and each sign product is weighted by the respective cell-wise missingness indicators $R_{ik}, R_{jk}, R_{il}, R_{jl}$.

Formally, the modified Kendall’s tau coefficient is defined as
\begin{equation}
\widetilde{\tau}_{kl} 
= \frac{2}{n(n-1)} \sum_{i=2}^n \sum_{j=1}^{i-1} 
\sgn(X_{ik}-X_{jk}) \, \sgn(X_{il}-X_{jl}) \,
R_{ik} R_{jk} R_{il} R_{jl},
\end{equation}
where $R_{ik} = 1$ if $X_{ik}$ is observed and $0$ otherwise, with $E(R_{ik}) = q_k$.
The test statistic is then given by
\begin{equation}
\widetilde{T}_{nd} = \sum_{k=2}^{d} \sum_{l=1}^{k-1} \widetilde{\tau}_{kl}^2.
\end{equation}

To facilitate the derivation of asymptotic properties, we also consider a $U$-statistic version of the test, denoted by $\widehat{T}_{nd}$, %in which $\widetilde{\tau}_{kl}^2$ is replaced by
$$\widehat{T}_{nd} = \sum_{1 \leq k \neq l \leq d} U_{kl},$$
where
\begin{equation}
U_{kl} \!=\! \frac{1}{n(n-1)(n-2)(n-3)} \sum_{1 \le i_1 \neq i_2 \neq i_3 \neq i_4 \le n} f_k(i_1,i_2) f_l(i_1,i_2) f_k(i_3,i_4) f_l(i_3,i_4),
\end{equation}
where $f_k(i,j) = \sgn(X_{ik}-X_{jk}) R_{ik} R_{jk}$. Under the null hypothesis, the statistic $\widehat{T}_{nd}$ is centered, that is, $E\widehat{T}_{nd} = 0.$

The $U$-statistic version $\widehat{T}_{nd}$ retains all available pairwise comparisons and offers a practical framework for deriving variance and asymptotic normality results, while the simpler statistic $\widetilde{T}_{nd}$ can be analyzed based on its relationship with $\widehat{T}_{nd}$.

\subsection{Asymptotic properties}

We study the asymptotic behavior of the proposed $U$-statistic $\widehat{T}_{nd}$ to provide a theoretical justification for its use in high-dimensional settings with missing data. Under the null hypothesis of mutual independence and assuming an MCAR mechanism and mutual independence between missingness indicators, with fixed column-wise completeness probabilities $q_1,\dots,q_d$, the standardized statistic is asymptotically normal.

For the cell-wise statistic $\widetilde{T}_{nd}$, we have explicitly derived both the expected value and the variance. The expected value can be written as
\begin{align*}
E (\widetilde{T}_{nd}) &= \sum_{k=2}^{d}\sum_{l=1}^{k-1} E\,\widetilde{\tau}_{kl}^2
= \frac{4}{n(n-1)}\sum_{k=2}^d \sum_{l=1}^{k-1} q_k^2q_l^2 \left(\frac{n-2}{9}q_kq_l + \frac{1}{2} \right),
\end{align*}
which provides a closed-form expression for the mean of the test statistic under the MCAR mechanism.
The variance of $\widetilde{T}_{nd}$ is more complex due to dependencies between overlapping pairs of indices. It can be decomposed into contributions from the individual terms and the covariances between them:
\begin{align*}
\operatorname{Var}&(\widetilde{T}_{nd}) = \operatorname{Var}\left( \sum_{k=2}^d \sum_{l=1}^{k-1} \widetilde{\tau}_{kl}^2 \right) = \sum_{k_1=2}^d \sum_{l_1=1}^{k_1-1}\sum_{k_2=2}^d \sum_{l_2=1}^{k_2-1} cov(\widetilde{\tau}_{k_1l_1}^2, \widetilde{\tau}_{k_2l_2}^2) \\ &= \sum_{k=2}^d \sum_{l=1}^{k-1} \operatorname{Var}(\widetilde{\tau}_{kl}^2) + \sum_{k=3}^d \sum_{l_1=1}^{k-1}\sum_{\substack{l_2=1 \\ l_2\neq l_1}}^{k-1} cov(\widetilde{\tau}_{kl_1}^2, \widetilde{\tau}_{kl_2}^2) + \sum_{k_1=2}^d \sum_{l=1}^{k_1-1}\sum_{\substack{k_2=l+1 \\ k_2 \neq k_1}}^d cov(\widetilde{\tau}_{k_1l}^2, \widetilde{\tau}_{k_2l}^2) \\ & + \sum_{k=2}^{d-1} \sum_{l_1=1}^{k-1}\sum_{k_2=k+1}^d cov(\widetilde{\tau}_{kl_1}^2, \widetilde{\tau}_{k_2k}^2) + \sum_{k_1=3}^d \sum_{l=2}^{k_1-1}\sum_{l_2=1}^{l-1} cov(\widetilde{\tau}_{k_1l}^2, \widetilde{\tau}_{ll_2}^2) .
\end{align*}
The variance of an individual squared term is given by
\begin{align*}
\operatorname{Var}(&\widetilde{\tau}_{kl}^2) \!=\! \frac{8 q_k^2 q_l^2}{2025 n^3(n - 1)^3 } \Big[2025 \!+\! 16650 (n \!-\! 2) q_k q_l \!+\! 675 (111 \!-\! 91n \!+\! 17n^2) q_k^2 q_l^2 \\ &\!+\! 18 (\!-\!3456 \!+\! 3694n \!-\! 1221n^2 \!+\! 119n^3) q_k^3 q_l^3 \\&+\! 50 (360 \!-\! 458n \!+\! 205n^2 \!-\! 37n^3 \!+\! 2n^4) q_k^4 q_l^4 \Big].
\end{align*}
Different covariance terms appear only when two squared statistics share at least one common index. These covariances have the same analytical structure and differ only by the specific indices involved. Thus, we give the explicit form for one case, noting that all others are obtained analogously. When two terms share the index $k$ the covariance between $\widetilde{\tau}_{kl_1}^2$ and $\widetilde{\tau}_{kl_2}^2$ is given by
\begin{align*}
\mathrm{cov}&(\widetilde{\tau}_{kl_1}^2, \widetilde{\tau}_{kl_2}^2) =\! \frac{8(1-q_k) q_k^2 q_{l_1}^2 q_{l_2}^2 }{27n^3(n-1)^3}\Big[27 \!+\! 9 q_k \left(-9 \!+\! 6n \!+\! 2(n \!-\! 2)(1 \!+\! (n \!-\! 2) q_k) q_{l_1}\right) \\ &\!+\! 2(n \!-\! 2) q_k \Big(9 \!+\! 9(n \!-\! 2) q_k \!+\! \left(2 \!+\! 4 q_k(\!-\!4 \!+\! 5 q_k) \!+\! 3n q_k(2 \!+\! (\!-\!5 \!+\! n) q_k)\right) q_{l_1} \Big) q_{l_2}\Big].
\end{align*}

Applying this reasoning, we obtain the variance of $\widehat{T}_{nd}$: 
\begin{align*}
\operatorname{Var}(\widehat{T}_{nd}) 
&= \frac{2}{n^2(n-1)^2(n-2)^2(n-3)^2}\sum_{1 \le k \neq l \le d}  \Bigl[
\frac{88 \, n (n-1)(n-2)(n-3)}{9} \, q_k^4 q_l^4 \\
&+ \frac{352 \, n (n-1)(n-2)(n-3)(n-4)}{75} \, q_k^5 q_l^5 + \frac{32 \, n (n-1)(n-2)(n-3)(n-4)(n-5)}{81} \, q_k^6 q_l^6
\Bigr].
\end{align*}
The following theorem formalizes the asymptotic normality of $\widehat{T}_{nd}$ under the null hypothesis and establishes the validity of the proposed test in high-dimensional settings.
 
\begin{theorem}\label{thm} 
Let us assume that we have an i.i.d. sample $\mathbf{X}_1,\dots, \mathbf{X}_n$ as $\mathbf{X}$, and let $\mathbf{R}_1,\dots, \mathbf{R}_n$ be the corresponding i.i.d. response indicator vectors with mutually independent components and the vector of completeness probabilities $(q_1, \dots, q_d)$ such that $0<q_{\min}\leq q_k, k = 1, \dots d$, for some $q_{\min}\in(0,1)$. We also assume the MCAR property. Then, under $H_0$, $\frac{\widehat{T}_{nd}}{\sqrt{\operatorname{Var}(\widehat{T}_{nd})}} \to N(0,1)$ as $d,n\to \infty$.
       % \item[(ii)] Under $H_0$, $\frac{\widehat{T}_{nd}}{\sqrt{\operatorname{Var}(\widehat{T}_{nd})}} \to N(0,1)$ as $d\to \infty$ for fixed $n$.
\end{theorem}

The proof of Theorem \ref{thm} is given in \ref{sec:proof}. The large values of $\frac{\widehat{T}_{nd}}{\sqrt{\operatorname{Var}(\widehat{T}_{nd})}}$ are considered significant.

In practice, the completeness probabilities $q_k$ are typically unknown and can be consistently estimated by
\[
\hat{q}_k = \frac{1}{n}\sum_{i=1}^n R_{ik},
\]
therefore the estimator of $\operatorname{Var}(\widehat{T}_{nd})$ is obtained via replacing $q_k$ with $\hat{q}_k$.

\begin{corollary}\label{cor:sigmahat}
Under the assumptions of Theorem~\ref{thm}, if $n\to\infty$ and $d\to\infty$, then under $H_0$ 
\[
\frac{\operatorname{\widehat{Var}}(\widehat{T}_{nd}) }{\operatorname{Var}(\widehat{T}_{nd}) }\overset{P}{\longrightarrow}1
\qquad\text{and}\qquad
\frac{\widehat{T}_{nd}}{\sqrt{\operatorname{\widehat{Var}}(\widehat{T}_{nd})}} \overset{D}{\longrightarrow}\mathcal{N}(0,1).
\]
\end{corollary}

\begin{lemma}\label{pom_lema}
    For all $x,y,x',y'\in[0,1]$ and $m\in\mathbb N$, $\bigl|x^my^m-x'^my'^m\bigr|\le m\bigl(|x-x'|+|y-y'|\bigr)$.
\end{lemma}

\begin{proof}
 Since $x'^m,y^m\le1$, $|x^my^m-x'^my'^m|\le|x^m -x'^m|\,y^m+x'^m\,|y^m-y'^m|\le|x^m-x'^m|+|y^m-y'^m|$. For $u,v\in[0,1]$, $|u^m-v^m|=|u-v|\,\bigl|\sum_{j=0}^{m-1}u^jv^{m-1-j}\bigr|\le m|u-v|$.
\end{proof}

\begin{proof}[Proof of Corollary~\ref{cor:sigmahat}.]
    Write $V_n=\frac{1}{2}\operatorname{Var}(\widehat{T}_{nd})n^2(n-1)^2(n-2)^2(n-3)^2$, $\widehat V_n=\frac{1}{2}\widehat{\operatorname{Var}}(\widehat{T}_{nd})n^2(n-1)^2(n-2)^2(n-3)^2$, $q_0=q_{\min}$ and $\Delta_k=\hat q_k-q_k$.

    By Lemma~\ref{pom_lema} with $m=4,5,6$ and the triangle inequality, \[
|\widehat V_n-V_n|\le\sum_{k\ne l}\bigl(4a_n+5b_n+6c_n\bigr)\bigl(|\Delta_k|+|\Delta_l|\bigr).
\]
The sum runs over ordered pairs, so $\sum_{k\ne l}(|\Delta_k|+|\Delta_l|)=2(d-1)\sum_{k=1}^d|\Delta_k|$. We have, 
$8a_n+10b_n+12c_n\le M n^6$ with $M=8\cdot\frac{88}{9}+10\cdot\frac{352}{75}+12\cdot\frac{32}{81}<130$. Hence,
\[
|\widehat V_n-V_n|\le M\,n^6\,(d-1)\sum_{k=1}^d|\hat q_k-q_k|.
\]

Since $\sum_{i}R_{ik}\sim\mathrm{Bin}(n,q_k)$, the Cauchy--Schwarz inequality gives
\[
E|\hat q_k-q_k|\le\sqrt{\operatorname{Var}(\hat q_k)}=\sqrt{\frac{q_k(1-q_k)}{n}}\le\frac1{2\sqrt n},
\]
uniformly in $k$. Therefore
\[
E|\widehat V_n-V_n|\le\frac{M}{2}\,n^{\frac{11}{2}}\,d(d-1).
\]

All summands of $V_n$ are nonnegative and $q_k\ge q_0$, so $V_n\ge c_n\,q_0^{12}\,d(d-1)$.
For $n\ge10$ we have $n-5\ge n/2$, hence $n(n-1)(n-2)(n-3)(n-4)(n-5)\ge(n/2)^6$ and $c_n\ge n^6/162$. Thus,
\( V_n\ge\frac{q_0^{12}}{162}\,n^6\,d(d-1).\) Finally, we have
\[
E\left|\frac{\widehat V_n}{V_n}-1\right|=\frac{E|\widehat V_n-V_n|}{V_n}\le81Mq_0^{-12}n^{-1/2}.
\]
By Markov's inequality, $\frac{\widehat V_n}{V_n}\xrightarrow{P}1$ as $n\to\infty$, uniformly in $d$. Since $|\sqrt x-1|\le|x-1|$ for $x\ge0$, also $\sqrt\frac{V_n}{\widehat V_n}\xrightarrow{P}1$. Writing
\[
\frac{\widehat T_{nd}}{\sqrt{\widehat V_n}}=\sqrt{\frac{V_n}{\widehat V_n}}\cdot\frac{\widehat{T}_{nd}}{\sqrt{V_n}}
\]
and using Theorem~\ref{thm} for the second factor, Slutsky's lemma gives $\frac{\widehat{T}_{nd}}{\sqrt{\widehat V_n}}\xrightarrow{D}N(0,1)$.
\end{proof}

The asymptotic normality of the standardized statistic $\widetilde{T}_{nd}$ can be derived under the additional constraint that the number of columns with missing observations $l$ remains fixed or grows slowly, i.e.\(l=o(\min\{n,d\})\). If missingness occurs in a larger fraction of columns, the asymptotic approximation for $\widetilde{T}_{nd}$ may no longer be valid, and the full $U$-statistic $\widehat{T}_{nd}$ should be used.

\section{Empirical study}\label{sec:empStudy}

We conduct a simulation study to investigate the finite-sample performance of the standardized version of $\widehat{T}_{nd}$, with estimated variance, in terms of empirical size and power at the nominal significance level of $0.05$.  The simulations are carried out for various combinations of the sample size $n\in\{10, 20, 40\}$ and the dimension $d\in\{4,8,16,32,64,128\}$. For each $(n,d)$ configuration, the results are based on $5000$ independent Monte Carlo replications.

We first examine the empirical size of the test under the null hypothesis of total independence. The components $X_{ik}$, $k = 1, \dots, d$ and $i = 1, \dots, n$, are generated independently from different marginal distributions exhibiting distinct tail behaviors, including the Gamma distribution $\gamma(4,2)$, the Student $t_4$ distribution and the standard Cauchy distribution.

To evaluate the impact of missing data, we consider column-wise missingness MCAR mechanisms with missingness rates equal to $0.1$ and $0.2$. We also consider the case where missingness rates differ across columns. The missingness rates are $0.1$, $0.2$, and $0.3$ for the first three columns, and $0.1$ for the remaining columns. The corresponding empirical sizes are reported in Tables~\ref{tab:size0.1}, ~\ref{tab:size0.2} and~\ref{tab: 3}, respectively. The total missingness rates, defined as the proportions of observations with at least one missing component, are summarized in the last row of each table. Note that a complete-case analysis would be highly inefficient in this high-dimensional setting, since even a relatively small proportion of missing values may result in most observations having at least one missing component, leaving only a small number of complete observations available for analysis.
\begin{table}[htbp]
    \centering
    \caption{Empirical size (\%): column-wise missingness rate {$0.1$}}
    \label{tab:size0.1}
    \begin{tabular}{cccccccc}
\hline
Marginal df.&$n \backslash d$ & 4    & 8    & 16   & 32   & 64 & 128\\
\hline
\multirow{3}{*}{$\gamma(4,2)$}&10 & 4.78 & 4.68 & 4.72 & 4.40 & 4.74 & 4.94 \\
&20 & 5.36 & 4.68 & 4.50 & 4.42 & 5.06 & 4.40 \\
&40 & 5.32 & 5.24 & 4.70 & 4.70 & 3.78 & 4.70 \\\hline
\multirow{3}{*}{$t_4$}&10 & 4.96 & 4.74 & 5.30 & 4.42 & 4.72 & 4.44 \\
&20 & 5.62 & 5.46 & 5.58 & 4.70 & 4.82 & 4.28 \\
&40 & 6.18 & 6.02 & 5.06 & 4.70 & 4.82 & 4.88\\\hline
\multirow{3}{*}{$\text{Cauchy}(0,1)$}&10 & 5.10 & 4.16 & 4.54 & 4.24 & 4.48 & 4.32 \\
&20 & 6.04 & 5.20 & 4.70 & 4.88 & 4.32 & 4.74 \\
&40 & 6.36 & 5.80 & 4.78 & 5.20 & 5.08 & 4.88\\
\hline
\multicolumn{2}{c}{total.miss.rate} & 0.35 & 0.57 & 0.81 & 0.97 & 1.00 & 1.00 \\
\hline
\end{tabular}
\end{table}

\begin{table}[hbtp]
    \centering
    \caption{Empirical size (\%): column-wise missingness rate {$0.2$}}
    \label{tab:size0.2}
    \begin{tabular}{cccccccc}
\hline
Marginal df.&$n \backslash d$ & 4    & 8    & 16   & 32   & 64  & 128 \\
\hline
\multirow{3}{*}{$\gamma(4,2)$}&10 & 3.10 & 2.96 & 2.84 & 2.70 & 2.84 & 3.42\\
&20 & 4.76 & 3.82 & 3.62 & 3.54 & 3.38 & 3.54 \\
&40 & 5.32 & 4.58 & 4.54 & 3.94 & 4.16 & 3.68 \\\hline
\multirow{3}{*}{$t_4$}&10 & 3.16 & 3.30 & 3.20 & 3.16 & 3.20 & 2.94\\
&20 & 5.12 & 4.20 & 4.10 & 3.66 & 3.40 & 3.76 \\ 
&40 & 5.68 & 5.62 & 4.52 & 4.24 & 4.40 & 4.26\\
\hline
\multirow{3}{*}{$\text{Cauchy}(0,1)$}&10 & 3.54 & 3.16 & 2.78 & 2.82 & 2.58 & 3.10 \\
&20 & 4.72 & 4.58 & 3.74 & 3.68 & 3.80 & 3.50 \\
&40 & 5.70 & 4.80 & 3.96 & 4.18 & 4.12 & 4.20\\
\hline
\multicolumn{2}{c}{total.miss.rate} & 0.59 & 0.83 & 0.97 & 0.99 & 1.00 & 1.00 \\
\hline
\end{tabular}
\end{table}

\begin{table}[hbtp]
    \centering
    \caption{Empirical size (\%): different column-wise missingness rates}
    \label{tab: 3}
    \begin{tabular}{cccccccc}
\hline
Marginal df.&$n \backslash d$ & 4    & 8    & 16   & 32   & 64 & 128\\
\hline
\multirow{3}{*}{$\gamma(4,2)$}&10 & 4.02 & 4.16 & 4.36 & 4.18 & 4.96 & 4.86 \\
&20 & 5.20 & 5.02 & 4.28 & 4.24 & 4.92 & 4.34\\
&40 & 5.38 & 5.26 & 4.84 & 4.92 & 3.80 & 4.10 \\

\hline
\multirow{3}{*}{$t_4$}&10 & 3.74 & 4.44 & 5.02 & 4.42 & 4.66 & 4.52  \\
&20 & 5.60 & 5.02 & 5.08 & 5.04 & 4.66 & 4.18 \\
&40 & 6.24 & 5.88 & 5.00 & 4.76 & 4.70 & 4.86 \\
 
\hline
\multirow{3}{*}{$\text{Cauchy}(0,1)$}&10 & 4.14 & 4.06 & 4.62 & 4.34 & 4.60 & 4.32  \\
&20 & 5.26 & 5.12 & 4.54 & 4.72 & 4.20 & 4.66  \\
&40 & 5.94 & 5.74 & 4.76 & 5.02 & 4.60 & 4.94\\
\hline
\multicolumn{2}{c}{total.miss.rate} &0.55 &0.70 &0.87 &0.98 &1.00 &1.00 \\
\hline
\end{tabular}
\end{table}

Overall, the empirical size of the proposed test remains close to the nominal level, even for small sample sizes. Mild deviations from the nominal level are observed primarily when the sample size is relatively small and fixed, while the dimension increases. These deviations may be attributed to finite-sample effects in the high-dimensional setting, which can become more pronounced at higher missingness rates. Nevertheless, the test shows good control of the empirical size even in high-dimensional settings and under heavy-tailed marginal distributions.

We next investigate the empirical power of the test under several dependence alternatives. The first alternative corresponds to a global dependence structure, where all components are mutually dependent. Specifically, the data are generated in the following way. Let $\Sigma_d = 0.95 I_d + 0.05 e_d e_d^T$, where $I_d$ denotes the $d$-dimensional identity matrix and $e_d = (1,\ldots,1)^T$. Observations are obtained as $X^T = \Sigma_d^{1/2} Z^T$, with $Z$ having independent components following a log-normal distribution $\mathrm{LN}(0,1)$.

The second alternative represents neighbor dependence, in which each variable depends only on its immediate neighbors. In this case, the data are generated as $X^T = 0.6 (Z_{i1},\ldots,Z_{ip})^T + 0.8 (Z_{i2},\ldots,Z_{i(p+1)})^T$, where $Z_{ik} \sim \mathrm{Cauchy}(0,1)$, introducing heavy-tailed marginals and local dependence.

Finally, we consider a multivariate normal alternative with a common pairwise correlation coefficient $\rho$, given by
\[
X \sim N_d(0,\Sigma), \qquad
\Sigma_{ij} =
\begin{cases}
1, & i = j, \\
\rho, & i \neq j .
\end{cases}
\]
This setting serves as a benchmark scenario with light-tailed distributions and a well-defined covariance structure.

The simulation study reveals distinct behaviors of the test under different dependence structures, with the corresponding results presented in Tables \ref{tab:power_mr01} and \ref{tab:power_mr02}. In the neighbor dependence (ND) model, the empirical power generally increases with the dimension, although the increase is less pronounced than under the global dependence (GD) model. This difference is consistent with the local nature of dependence in the ND model, in contrast to the GD model, where dependence is present between all pairs of variables. In contrast, under the global dependence (GD) model, the power increases noticeably as $d$ grows, since dependence is present between all pairs of variables, making it easier for the test to detect the overall signal as the dimensionality increases. Under the multivariate normal alternative with pairwise correlation $\rho=0.3$, the empirical power of the test increases both with the sample sizes $n$ and the dimension $d$. 
For small $n$, the power ranges from moderate to high values depending on $d$, while for larger sample sizes it approaches or reaches $100\%$ across all dimensions. Overall, the proposed test captures dependence effectively, with its sensitivity depending on the structure of the underlying relationship.

\begin{table}[hbtp]
    \centering
    \caption{Empirical power (\%): column-wise missingness rate {$0.1$}}
    \label{tab:power_mr01}
    \begin{tabular}{cccccccc}
\hline
Model&$n \backslash d$ & 4    & 8    & 16   & 32   & 64 & 128   \\
\hline
\multirow{3}{*}{ND}&10 & 23.00 & 23.46 & 24.18 & 26.14 & 25.82 & 26.14\\
&20 & 60.48 & 69.48 & 75.86 & 79.32 & 80.84 & 82.36\\
&40 & 96.70 & 99.68 & 99.96 & 100 & 100 & 100\\\hline
\multirow{3}{*}{GD}&10 & 5.70 & 6.88 & 12.08 & 26.24 & 56.14 & 82.74 \\
&20 & 8.06 & 13.10 & 29.00 & 65.14 & 93.58 & 99.68\\ 
&40 & 13.60 & 29.30 & 65.58 & 95.40 & 99.92 & 100\\\hline
\multirow{3}{*}{$\mathcal{N}_{0.3}$}&10 & 17.72 & 34.44 & 59.14 & 81.06 & 93.08 & 97.98 \\
&20 & 42.16 & 72.90 & 93.72 & 99.38 & 99.98 & 100 \\ 
&40 & 77.28 & 97.76 & 99.92 & 100 & 100 & 100 \\\hline
\multicolumn{2}{c}{total.miss.rate} & 0.35 & 0.57 & 0.81 & 0.97 & 1.00 & 1.00 \\
\hline
\end{tabular}
\end{table}

\begin{table}[htbp]
\centering
\caption{Empirical power (\%): under column-wise missingness rate 0.2.}
\label{tab:power_mr02}
\begin{tabular}{ccrrrrrr}
\hline
Model & $n \backslash d$ & 4 & 8 & 16 & 32 & 64 & 128 \\
\hline
\multirow{3}{*}{ND}
 & 10 & 13.48 & 12.82 & 12.70 & 13.84 & 13.44 & 13.86 \\
 & 20 & 41.50 & 47.18 & 51.28 & 54.54 & 56.10 & 57.14 \\
 & 40 & 86.84 & 95.30 & 98.64 & 99.34 & 99.60 & 100 \\
\hline
\multirow{3}{*}{GD}
 & 10 & 3.78 & 4.30 & 6.88 & 15.26 & 40.44 & 70.52 \\
 & 20 & 6.04 & 9.08 & 19.28 & 50.00 & 85.76 & 98.54 \\
 & 40 & 10.70 & 20.50 & 51.76 & 89.32 & 99.76 & 100 \\
\hline
\multirow{3}{*}{$\mathcal{N}_{0.3}$}
 & 10 & 9.92 & 21.18 & 41.44 & 67.76 & 87.00 & 95.52 \\
 & 20 & 29.24 & 58.30 & 86.64 & 97.84 & 99.90 & 100 \\
 & 40 & 65.02 & 93.62 & 99.66 & 100 & 100 & 100 \\
\hline
\multicolumn{2}{c}{total.miss.rate} & 0.59 & 0.83 & 0.97 & 0.99 & 1.00 & 1.00 \\
\hline
\end{tabular}
\end{table}

\subsection{Missing at random (MAR)}

Although the asymptotic results are derived under the MCAR assumption, it is of practical interest to investigate their applicability when this assumption does not hold. Following \cite{aleksic2025two}, we consider the MAR mechanism 1–9 implemented in the R package missMethods \cite{missMethods}. In all considered MAR scenarios, the first column is used as a fully observed reference variable, i.e., no missingness is generated in this column, while the remaining columns are subject to missingness with probability 0.1 and 0.2. All other simulation settings are kept the same as in the previous section. For comparison purposes we also present the results for MCAR mechanism with the same column missingness rates as in the considered MAR scenario. The probabilities of type I errors are presented in Tables \ref{tab:sizeMAR} and \ref{tab:sizeMAR1}. It is evident that the empirical sizes are close across different mechanisms and close to the nominal values even for small sample sizes.  From Table \ref{tab:power_mr01_MAR}, we observe that for a small column-wise missingness rate of 0.1, the empirical powers against the previously considered alternatives are quite similar. In contrast, when the column-wise missingness rate increases to 0.2, the empirical powers under the MAR mechanism become larger, with the impact of the missingness mechanism depending on the specific alternative and being most pronounced for the ND alternative. It is also worth noting that the largest differences in power occur for 
$n=20$, reflecting the fact that in this setting the sample size is neither sufficiently large nor sufficiently small relative to the dimension $d$.

\begin{table}[htbp]
\centering
\caption{Empirical size (\%): under MCAR and MAR, with missingness rate $0.1$ in all columns except the first, which is fully observed}
\label{tab:sizeMAR}
\begin{tabular}{ccccccccc}
\hline
Model & Missingness & $n \backslash d$ & 4 & 8 & 16 & 32 & 64 & 128 \\
\hline
\multirow{6}{*}{$\gamma(4,2)$}
 & \multirow{3}{*}{MCAR}
 & 10 & 5.22 & 4.46 & 5.16 & 4.56 & 4.60 & 4.74\\
 &  & 20 & 5.54 & 5.64 & 4.74 & 5.04 & 4.28 & 4.44 \\
 &  & 40 & 5.82 & 4.50 & 5.00 & 4.64 & 4.72 & 4.76\\
\cline{2-9}
 & \multirow{3}{*}{MAR}
 & 10 & 4.58 &4.28& 4.52& 4.04& 4.36& 4.18
 \\
 &  & 20 & 5.76 & 5.06 & 4.94 & 5.60 & 4.70 & 5.18\\
 &  & 40 & 5.90 & 5.78 & 5.46 & 5.32 & 5.00 & 5.28 \\
\hline

\multirow{6}{*}{$t_4$}
 & \multirow{3}{*}{MCAR}
 & 10 & 5.34 &  5.40 & 4.40 &  4.44 & 4.36 & 4.50
 \\
 &  & 20 & 6.18 &  5.30 & 5.02 &  4.50 & 4.98 & 4.68  \\
 &  & 40 & 6.34 & 5.54 & 5.32 & 5.20 & 5.04 & 4.92 \\
\cline{2-9}
 & \multirow{3}{*}{MAR}
 & 10 & 5.20 & 4.28 & 4.52 & 4.14 & 4.28& 4.14
 \\
 &  & 20 & 5.52 & 5.58 & 4.80 & 5.16 & 4.84 & 4.94 \\
 &  & 40 & 6.68 & 5.84 & 5.30 & 5.28 & 5.42 & 5.24  \\
 \hline

\multirow{6}{*}{$\text{Cauchy}(0,1)$}
& \multirow{3}{*}{MCAR}
 & 10 & 5.74 & 4.92 & 4.80 & 4.74 & 4.60 & 4.18  \\
 &  & 20 &  5.48 & 5.88 & 4.80 & 5.02 & 4.78 & 4.44 \\
 &  & 40 & 6.64	& 5.86 & 5.36 &	4.88 & 4.60 & 4.28 \\
\cline{2-9}
 & \multirow{3}{*}{MAR}
 & 10 & 5.10 & 4.98 & 4.68 & 4.66 & 4.16 & 4.30 \\
 &  & 20 & 5.92 & 6.00 & 5.46 & 5.04 & 5.54 & 4.88  \\
 &  & 40 & 6.30	& 6.08 & 5.48 & 5.00 & 5.32 & 4.80 \\
\hline

\multicolumn{2}{c}{total.miss.rate} & & 0.27 & 0.52 & 0.79 & 0.96  & 1.00  & 1.00 \\
\hline
\end{tabular}
\end{table}

\begin{table}[htbp]
\centering
\caption{Empirical size (\%): under MCAR and MAR, with missingness rate $0.2$ in all columns except the first, which is fully observed}
\label{tab:sizeMAR1}
\begin{tabular}{ccccccccc}
\hline
Model & Missingness & $n \backslash d$ & 4 & 8 & 16 & 32 & 64 & 128 \\
\hline
\multirow{6}{*}{$\gamma(4,2)$}
 & \multirow{3}{*}{MCAR}
 & 10 &  4.56 & 3.76 & 3.42 & 3.04  & 3.00 & 2.78 \\
 &  & 20 & 5.28 & 4.96 & 3.98 & 4.08 & 3.76 & 3.70 \\
 &  & 40 & 5.96 & 5.50 & 4.26 & 4.32 & 4.46 & 4.20 \\
\cline{2-9}
 & \multirow{3}{*}{MAR}
 & 10 & 3.74 & 3.52 & 4.12 & 3.54 & 3.60 & 4.06  \\
 &  & 20 & 5.58 &  5.66 & 5.46 &	5.84 & 5.98	& 5.72 \\
 &  & 40 & 5.82 & 5.36 &	5.48 & 4.80 & 5.74 & 5.80 \\
\hline

\multirow{6}{*}{$t_4$}
 & \multirow{3}{*}{MCAR}
 & 10 & 4.42 & 4.20 & 3.40 & 3.38 & 2.98 & 3.56    \\
 &  & 20 & 5.98 & 4.28 & 4.18 & 3.50 & 3.56 & 3.44  \\
 &  & 40 & 6.24 & 5.72 & 4.68 & 4.30 & 4.66 & 3.68 \\
\cline{2-9}
 & \multirow{3}{*}{MAR}
 & 10 & 3.96 & 3.78 & 3.48 & 3.66 & 3.60 & 3.90   \\
 &  & 20 & 5.70 & 5.68 & 5.72 & 5.20  &5.76 &5.62  \\
 &  & 40 & 5.98	& 6.06	& 5.96	& 5.82 & 5.74 & 5.86 \\
 \hline

\multirow{6}{*}{$\text{Cauchy}(0,1)$}
 & \multirow{3}{*}{MCAR}
 & 10 & 4.56 & 3.88 & 4.02 & 3.46 & 3.46 & 2.96    \\
 &  & 20 & 5.02 & 4.52 & 4.66 & 3.86 & 3.88 & 3.34 \\
 &  & 40 & 6.24	& 5.72 &4.68 & 4.30 & 4.06 & 4.02\\
\cline{2-9}
 & \multirow{3}{*}{MAR}
 & 10 & 4.00 & 4.24 & 4.02 & 4.54 & 3.92 & 4.30  \\
 &  & 20 & 5.90 & 6.24 & 5.82 & 5.76 & 6.22 & 5.56  \\
 &  & 40 & 6.32 & 5.96 & 5.74 & 5.20 & 5.46 & 5.50  \\
\hline
\multicolumn{2}{c}{total.miss.rate} & & 0.49 & 0.79& 0.96 & 0.99  & 1.00  & 1.00 \\
\hline
\end{tabular}
\end{table}

\begin{table}[hbtp]
    \centering
    \caption{Empirical power (\%): under MCAR and MAR, with missingness rate $0.1$ in all columns except the first, which is fully observed}
    \label{tab:power_mr01_MAR}
    \begin{tabular}{ccccccccc}
\hline
Model& Missingness &$n \backslash d$ & 4    & 8    & 16   & 32   & 64 & 128   \\
\hline
\multirow{6}{*}{ND}
&\multirow{3}{*}{MCAR}&10& 23.84&23.96&23.98&26.16&26.04&25.62\\
& &20 &61.20&69.91&75.46&78.18&81.32&81.18\\
& &40 &96.26&99.56&99.94& 100 & 100 & 100 \\
\cline{2-9}
&\multirow{3}{*}{ MAR}&10 &23.76&24.76&24.30&25.28&26.06&25.68  \\
& &20 &63.40&71.80&77.82&80.66&82.60&83.54\\ 
& &40 &97.18&99.68&100 & 100 & 100 & 100 \\
\hline
\multirow{6}{*}{GD}
&\multirow{3}{*}{MCAR}&10&6.80 & 7.34 & 12.42 & 27.20 & 55.26 & 83.28\\
& &20 &8.80 & 14.00 & 30.02 & 63.88 & 93.10 & 99.66\\
& &40 &14.64 & 29.90 & 66.22 & 96.20 & 99.94 & 100 \\
\cline{2-9}
&\multirow{3}{*}{ MAR}&10 &6.04 & 7.42 & 12.08 & 26.42 & 54.50 & 81.90\\
& &20 &9.16 & 14.68 & 31.44 & 65.42 & 93.36 & 99.54\\ 
& &40 &15.52 & 31.20 & 66.20 & 96.00 & 99.94 & 100\\
\hline
\multirow{6}{*}{$\mathcal{N}_{0.3}$}
&\multirow{3}{*}{MCAR}&10&18.98 & 35.32 & 60.24 & 81.26 & 92.90 & 98.02 \\
& &20 &45.64 & 74.44 & 93.90 & 99.32 & 99.98 & 100 \\
& &40 &79.80 & 98.02 & 99.98 & 100 & 100 & 100 \\

\cline{2-9}
&\multirow{3}{*}{ MAR}&10 &18.82 & 35.74 & 59.62 & 81.58 & 93.02 & 97.94 \\
& &20 &46.14 & 75.42 & 94.22 & 99.46 & 99.98 & 100 \\ 
& &40 &81.48 & 98.26 & 99.96 & 100 & 100 & 100\\
\hline
\multicolumn{2}{c}{total.miss.rate} & & 0.27 & 0.52 & 0.79 & 0.96  & 1.00  & 1.00 \\
\hline
\end{tabular}
\end{table}

\begin{table}[hbtp]
    \centering
    \caption{Empirical power (\%): under MCAR and MAR, with missingness rate $0.2$ in all columns except the first, which is fully observed}
    \label{tab:power_mr02_MAR}
    \begin{tabular}{ccccccccc}
\hline
Model& Missingness &$n \backslash d$ & 4    & 8    & 16   & 32   & 64 & 128   \\
\hline
\multirow{6}{*}{ND}
&\multirow{3}{*}{MCAR}&10& 15.84 &13.62& 14.06& 14.96& 14.46&13.64\\
& &20 & 44.34& 46.44& 50.50& 53.66& 54.08& 57.00\\
& &40 & 85.38 & 94.56 & 98.06 &99.22 & 99.64 & 100\\
\cline{2-9}
&\multirow{3}{*}{ MAR}&10& 15.56& 16.12& 16.40& 17.38& 16.50& 16.98\\
& &20 &48.54& 56.00& 61.84& 65.92& 68.70& 71.12 \\
& &40 &90.32 & 97.24 & 99.18 & 99.80 & 99.89 & 100\\
\hline
\multirow{6}{*}{GD}
&\multirow{3}{*}{MCAR}&10&5.26 & 5.44 & 8.02 & 18.26 & 40.56 & 70.98\\
& &20 &7.42 & 9.76 & 20.76 & 49.30 & 86.58 & 98.66\\
& &40 &12.22 & 22.58 & 53.74 & 90.68 & 99.76 & 100\\
\cline{2-9}
&\multirow{3}{*}{ MAR}&10&4.46 & 5.58 & 9.38 & 19.08 & 41.44 & 71.22\\
& &20 &8.18 & 12.96 & 25.36 & 55.06 & 87.78 & 98.66\\
& &40 &12.94 & 25.14 & 55.52 & 91.44 & 99.66 & 100\\
\hline
\multirow{6}{*}{$\mathcal{N}_{0.3}$}
&\multirow{3}{*}{MCAR}& 10 & 14.16 & 24.72 & 44.86 & 68.66 & 87.04 & 95.44 \\
& & 20 & 36.24 & 62.94 & 87.48 & 97.88 & 99.90 & 100\\
& & 40 & 71.22 & 95.20 & 99.86 & 100 & 100 & 100 \\
\cline{2-9}
&\multirow{3}{*}{ MAR}&10 & 13.60 & 25.82 & 46.30 & 70.86 & 87.06 & 95.42\\
& &20 & 38.20 & 65.72 & 89.16 & 98.44 & 99.90 & 100\\
& &40 & 72.40 & 95.72 & 99.90 & 100 & 100 & 100\\
\hline
\multicolumn{2}{c}{total.miss.rate} & & 0.49 & 0.79 & 0.96 & 0.99 & 1.00 & 1.00 \\
\hline
\end{tabular}
\end{table}

\section{Conclusion}\label{sec:conclusion}

To test total independence in high-dimensional data in the presence of missingness, without assuming Gaussianity, we propose, to the best of our knowledge for the first time, an adaptation of the test studied by \cite{mao2018testing}. The proposed procedure provides a feasible alternative to the complete-case approach, which is not applicable in high-dimensional settings. Under the MCAR mechanism, we derive the asymptotic null distribution of the test statistic and use it for p-value approximation. 
The simulation results also suggest that the proposed test maintains satisfactory empirical size under the considered MAR mechanism, although the theoretical results are established under MCAR.

In light of the derived results, future research may focus on extending the proposed framework to MAR mechanisms in a formal manner, accompanied by appropriate theoretical guarantees, as well as on developing tests for the MCAR assumption in high-dimensional settings.

\section*{Acknowledgments}

The work of B. Milošević and M. Cuparić is supported by the Ministry of Science, Technological Development and Innovations of the Republic of Serbia (contract no. 451-03-33/2026-03/200104), while the work of J. Radojević is supported by the Ministry of Science, Technological Development and Innovations of the Republic of Serbia (contract no. 200092). The results are obtained upon work from COST Action HiTEc-Text, functional and other high-dimensional data in econometrics: New models, methods, applications, CA21163, supported by COST (European Cooperation in Science and Technology).

\section*{Declaration on AI usage}
AI tools were used for grammar and language editing and code formatting. The authors take full responsibility for the accuracy of the work presented.

%\appendix

\section*{Appendix: Proof of Theorem \ref{thm}}\label{sec:proof}
Notice that the $U$-statistic $\widehat{T}_{nd}$ can be written as a sum over all pairs of variables:
$$\widehat{T}_{nd} = \sum_{1 \leq j_1 \neq j_2 \leq d} U_{j_1j_2.}$$
Let $\mathcal{F}_{n,k}=\sigma\{(X_{ij}, R_{ij}) \, |\,i=1,2...,n, j=1,2...,k \}$ denote the $\sigma$-field, and define $Y_{n,k}=\frac{2}{\sqrt{\operatorname{Var}(\widehat{T}_{nd})}}\sum_{l=1}^{k-1}U_{kl}$, where
\begin{align*}
          U_{kl} = \frac{1}{n(n-1)(n-2)(n-3)} \sum_{1\leq i_1 \neq i_2 \neq i_3\neq i_4 \leq n} f_{k}(i_1, i_2)f_{l}(i_1, i_2)f_{k}(i_3, i_4)f_{l}(i_3, i_4) 
    \end{align*}
and $f_{k}(i, j) = \sgn(X_{ik}-X_{jk})R_{ik}R_{jk}$. Note that 
\begin{align}\label{TnYn}
   \frac{\widehat{T}_{nd}}{\sqrt{\operatorname{Var}(\widehat{T}_{nd})}} &=  \frac{2}{\sqrt{\operatorname{Var}(\widehat{T}_{nd})}}\sum_{j_2 = 2}^d \sum_{j_1 = 1}^{j_2-1}  U_{j_1j_2} = \sum_{j_2=2}^d Y_{n, j_2}. 
\end{align}
Before proceeding to the main theorem, we first establish an important property of the ${Y_{n,k}}$.

\begin{lemma}\label{lemma_martingale}
    Under $H_0$, $\{Y_{n,k},\mathcal{F}_{n,k}\,,\,2\leq k \leq d, n\geq4\}$ is a martingale difference (m.d.) array.
\end{lemma} 

\begin{proof}

We show that the martingale difference property holds when $n$ varies, i.e. that for all $m \leq n-1$ and $t \leq k-1$, $$E(Y_{n,k}| \mathcal{F}_{m,t}) = 0.$$
Indeed,
\begin{align*}
    E(&U_{kl}| \mathcal{F}_{m,t}) = \sum_{1\leq i_1 \neq i_2 \neq i_3\neq i_4 \leq n} E(f_{k}(i_1, i_2)f_{l}(i_1, i_2)f_{k}(i_3, i_4)f_{l}(i_3, i_4 )|\mathcal{F}_{m,t}) \\ &= \sum_{1\leq i_1 \neq i_2 \neq i_3\neq i_4 \leq n} E(E(f_{k}(i_1, i_2)f_{l}(i_1, i_2)f_{k}(i_3, i_4)f_{l}(i_3, i_4 )|\mathcal{F}_{n,k-1})|\mathcal{F}_{m,t})) \\ &= \sum_{1\leq i_1 \neq i_2 \neq i_3\neq i_4 \leq n} E(E(f_{k}(i_1, i_2)f_{k}(i_3, i_4)E(f_{l}(i_1, i_2)f_{l}(i_3, i_4 )|\mathcal{F}_{n,k-1})|\mathcal{F}_{m,t}))
    \\ &= 0,
\end{align*}
where the last equality holds because \[
\begin{aligned}
&E[f_k(i_1,i_2)f_l(i_1,i_2)
f_k(i_3,i_4)f_l(i_3,i_4)
\mid\mathcal F_{n,k-1}]=f_l(i_1,i_2)f_l(i_3,i_4)
E[f_k(i_1,i_2)f_k(i_3,i_4)]=0.
\end{aligned}
\] under $H_0$. This completes the proof of the lemma.
\end{proof}

Since \eqref{TnYn} is a sum of m.d. array or m.d. sequence, we can use related central limit theorems to prove Theorem \ref{thm}. To that end, we first state several technical results

\begin{lemma}\label{lemma_cond}
    Under $H_0$, 
\begin{itemize}
    \item[(i)] $\sum_{k=2}^dEY^4_{n,k}\to 0$, and 
    \item[(ii)] $\sum_{k=2}^dY^2_{n,k}\overset{P}{\to} 1$ as $d, n \to \infty$,
\end{itemize}
%irrespective of whether $n$ is diverging or fixed, 
where $\xrightarrow{P}$ denotes convergence in probability.
\end{lemma}

\begin{proof}
First, we analyze the fourth moment of $Y_{n,j_2}$. Recall that $$ EY_{n,j_2}^4 = \frac{16}{(\operatorname{Var}(\widehat{T}_{nd}))^2} \!\sum_{1\leq l_1,l_2,l_3,l_4 \leq j_2-1}\! E U_{j_2l_1}U_{j_2l_2}U_{j_2l_3}U_{j_2l_4}.$$
To evaluate this expression, we first consider the expectation of a term in the sum. Let $A = \sigma(X_{i{j_2}}, R_{i{j_2}}, 1\leq i \leq n)$. For mutually distinct $l_1, l_2, l_3$ and $l_4$, we have
   \begin{align}\label{ineq l1l2l3l4}
        &E(U_{j_2l_1}U_{j_2l_2}U_{j_2l_3}U_{j_2l_4})\!=\! E(E(U_{j_2l_1}U_{j_2l_2}U_{j_2l_3}U_{j_2l_4}|A))\nonumber \\&=\! \frac{1}{n^4(n\!-\!1)^4(n\!-\!2)^4(n\!-\!3)^4} \nonumber\\&\times E\sum_{\substack{1\leq i_1 \neq i_2 \neq i_3 \neq i_4 \leq n \\ 1\leq i_1^{'} \neq i_2^{'} \neq i_3^{'} \neq i_4^{'} \leq n \\ 1\leq i_1^{''} \neq i_2^{''} \neq i_3^{''} \neq i_4^{''} \leq n \\ 1\leq i_1^{'''} \neq i_2^{'''} \neq i_3^{'''} \neq i_4^{'''} \leq n }} 
        \begin{aligned}[t]
        &f_{j_2}(i_1, i_2)f_{j_2}(i_3, i_4)f_{j_2}(i_1^{'}, i_2^{'})f_{j_2}(i_3^{'}, i_4^{'})\\ &f_{j_2}(i_1^{''}, i_2^{''})f_{j_2}(i_3^{''}, i_4^{''}) f_{j_2}(i_1^{'''}, i_2^{'''})f_{j_2}(i_3^{'''}, i_4^{'''}) \\ &E[ f_{l_1}(i_1, i_2)f_{l_1}(i_3, i_4)|A] \cdot\dots\cdot E[ f_{l_4}(i_1^{'''}, i_2^{'''})f_{l_4}(i_3^{'''}, i_4^{'''})|A].
        \end{aligned}
   \end{align}
Since $E[ f_{l_1}(i_1, i_2)f_{l_2}(i_3, i_4)|A] = Ef_{l_1}(i_1, i_2)f_{l_2}(i_3, i_4) = Ef_{l_1}(i_1, i_2)Ef_{l_2}(i_3, i_4) = 0$, it follows that \eqref{ineq l1l2l3l4} is equal to $0$. In the same way, it is easy to see that $E(U_{j_2l_1}U_{j_2l_2}U_{j_2l_3}U_{j_2l_4}) = 0$ whenever exactly two or three indices $l_1, l_2, l_3$ and $l_4$ are equal. For the case $l_1 = l_2 = l_3=l_4$,
\begin{align*}
    E(&U_{j_2l_1}^4) = \frac{1}{n^4(n-1)^4(n-2)^4(n-3)^4} \nonumber\\&\times \sum_{\substack{1\leq i_1 \neq i_2 \neq i_3 \neq i_4 \leq n \\ 1\leq i_1^{'} \neq i_2^{'} \neq i_3^{'} \neq i_4^{'} \leq n \\ 1\leq i_1^{''} \neq i_2^{''} \neq i_3^{''} \neq i_4^{''} \leq n \\ 1\leq i_1^{'''} \neq i_2^{'''} \neq i_3^{'''} \neq i_4^{'''} \leq n }} 
        \begin{aligned}[t]
        &E[f_{j_2}(i_1, i_2)f_{j_2}(i_3, i_4)f_{j_2}(i_1^{'}, i_2^{'})f_{j_2}(i_3^{'}, i_4^{'}) \\ &f_{j_2}(i_1^{''}, i_2^{''})f_{j_2}(i_3^{''}, i_4^{''}) f_{j_2}(i_1^{'''}, i_2^{'''})f_{j_2}(i_3^{'''}, i_4^{'''})] \\ &E[f_{l_1}(i_1, i_2)f_{l_1}(i_3, i_4)f_{l_1}(i_1^{'}, i_2^{'})f_{l_1}(i_3^{'}, i_4^{'}) \\ &f_{l_1}(i_1^{''}, i_2^{''})f_{l_1}(i_3^{''}, i_4^{''})f_{l_1}(i_1^{'''}, i_2^{'''})f_{l_1}(i_3^{'''}, i_4^{'''})].
        \end{aligned}
\end{align*}

If one of the eight comparison pairs shares no observation
index with any of the remaining pairs, the corresponding
factor is independent of the other factors and has zero
expectation under $H_0$. Therefore, for a summand to be
nonzero, each comparison pair must share at least one
observation index with another pair. Consequently, the eight pairs can be partitioned into
at most four connected groups, each containing at least
two pairs. Since a connected group of $r$ pairs involves
at most $r+1$ distinct indices, the total number of
distinct observation indices is at most $8+4=12$. Hence, the number of potentially nonzero summands is
$O(n^{12})$.
Moreover, since the kernels are uniformly bounded,
there exists a constant $C>0$, independent of
$n$, $k$, and $l$, such that
\begin{align*}
    \Big | &E[f_{j_2}(i_1, i_2)f_{j_2}(i_3, i_4)f_{j_2}(i_1^{'}, i_2^{'})f_{j_2}(i_3^{'}, i_4^{'}) f_{j_2}(i_1^{''}, i_2^{''})f_{j_2}(i_3^{''}, i_4^{''}) f_{j_2}(i_1^{'''}, i_2^{'''})f_{j_2}(i_3^{'''}, i_4^{'''})] \\ &\times E[f_{l_1}(i_1, i_2)f_{l_1}(i_3, i_4)f_{l_1}(i_1^{'}, i_2^{'})f_{l_1}(i_3^{'}, i_4^{'}) f_{l_1}(i_1^{''}, i_2^{''})f_{l_1}(i_3^{''}, i_4^{''})f_{l_1}(i_1^{'''}, i_2^{'''})f_{l_1}(i_3^{'''}, i_4^{'''})] \Big | \leq C.
\end{align*}
It follows that $E (U_{j_2l_1}^{4})
\le \frac{Cn^{12}}{n^4(n-1)^4(n-2)^4(n-3)^4}$. We conclude that $E (U_{j_2l_1}^{4}) = O(n^{-4})$.

Now consider the case $l_1=l_2\neq l_3=l_4$.

By the Cauchy--Schwarz inequality, we have
\begin{align}\label{eq: proizvodKvadrata}
E\left(
U_{j_2l_1}^{2}U_{j_2l_3}^{2}
\right)
&\le
\left({E}U_{j_2l_1}^{4}\right)^{\frac{1}{2}}
\left(EU_{j_2l_3}^{4}\right)^{\frac{1}{2}}\le \frac{C_1}{n^4},
\end{align}
where the last inequality follows from the
previously established uniform fourth-moment
bound. Consequently, $E(U_{j_2l_1}^2 U_{j_2l_3}^2)  = O(n^{-4})$.

%Similarly, this last expression can also be bounded by $O\big(\frac{1}{n^6}\big)$, since only a limited number of index combinations give a nonzero contribution. 
As a consequence,
\begin{align*}
     \sum_{j_2=2}^{d}EY_{n,j_2}^4 &= \frac{16}{(\operatorname{Var}(\widehat{T}_{nd}))^2} \sum_{j_2=2}^{d}\sum_{1\leq l_1,l_2,l_3,l_4 \leq j_2-1} E(U_{j_2l_1}U_{j_2l_2}U_{j_2l_3}U_{j_2l_4}) \\ &\!=\! \frac{16\cdot 3}{(\operatorname{Var}(\widehat{T}_{nd}))^2} \sum_{j_2=2}^{d}\sum_{l_1 = 1}^{j_2-1}\sum_{\substack{l_3 = 1 \\ l_3 \neq l_1}}^{j_2-1}E(U_{j_2l_1}^2U_{j_2l_3}^2) \!+\! \frac{16}{(\operatorname{Var}(\widehat{T}_{nd}))^{2}} \sum_{j_2=2}^{d}\sum_{l_1 = 1}^{j_2-1}EU_{j_2l_1}^4 \\ &=\!  \frac{16\cdot 3}{(\operatorname{Var}(\widehat{T}_{nd}))^2}O\left(\frac{1}{n^4}\right)\! \sum_{j_2=2}^{d} \!(j_2\!-\!1)(j_2\!-\!2) \!+\! \frac{16}{(\operatorname{Var}(\widehat{T}_{nd}))^2} O\left(\frac{1}{n^4}\right)\!\sum_{j_2=2}^{d} (j_2\!-\!1) \\ &= \frac{16}{(\operatorname{Var}(\widehat{T}_{nd}))^2}O\left(\frac{1}{n^4}\right) d(d-1)(d-2) + \frac{16}{(\operatorname{Var}(\widehat{T}_{nd}))^2} O\left(\frac{1}{n^4}\right)\frac{d(d-1)}{2} \\ &= O\left(\frac{d(d-1)(d-2)+d(d-1)}{d^2(d-1)^2}\right) = O\left(\frac{1}{d}\right)
\end{align*}

Therefore, $ \sum_{j_2=2}^{d}EY_{n,j_2}^4 \to 0$ as $ d \to \infty$.

To show that $\sum_{k=2}^{d}Y_{n,k}^2 \xrightarrow{P} 1$, it suffices to prove that $E\left(\sum_{k=2}^{d}Y_{n,k}^2 - 1\right)^2 \to 0$ as $d \to \infty$.
Note that \begin{align*}
    E\left(\sum_{k=2}^{d}Y_{n,k}^2 - 1\right)^2 % &= E\Big(\sum_{k_1=2}^d Y_{n,k_1}^2 \sum_{k_2=2}^d Y_{n,k_2}^2\Big)-2\sum_{k=2}^d EY_{n,k}^2 + 1 \\ &
    = \sum_{k_1=2}^d \sum_{k_2=2}^d E(Y_{n,k_1}^2Y_{n,k_2}^2) - 2\sum_{k=2}^d EY_{n,k}^2 + 1.
\end{align*}
Further, we have
\begin{align*}
    \sum_{k=2}^d EY_{n,k}^2 &\!=\! \frac{4}{\operatorname{Var}(\widehat{T}_{nd})} \sum_{k=2}^d E\left(\sum_{l=1}^{k-1}U_{kl}\right)^2 \!=\! \frac{4}{\operatorname{Var}(\widehat{T}_{nd})} \sum_{k=2}^d\sum_{l_1=1}^{k-1}\sum_{l_2=1}^{k-1}E(U_{kl_1}U_{kl_2}) \\ &= \frac{4}{\operatorname{Var}(\widehat{T}_{nd})} \sum_{k=2}^d\sum_{\substack{l_1, l_2=1 \\ l_1 \neq l_2} }^{k-1}E(U_{kl_1}U_{kl_2}) + \frac{4}{\operatorname{Var}(\widehat{T}_{nd})} \sum_{k=2}^d\sum_{l=1}^{k-1}E(U_{kl}^2) \\ &= 0 + \frac{4}{\operatorname{Var}(\widehat{T}_{nd})} \sum_{k=2}^d\sum_{l=1}^{k-1}\operatorname{Var}(U_{kl}) = 1. %\\ &= 4\cdot\frac{1}{4} = 1.
\end{align*}
Hence, we have
\begin{align*}
    E\left(\sum_{k=2}^{d} Y_{n,k}^2 - 1\right)^2 = \sum_{k=2}^{d} E(Y_{n,k}^4) + \sum_{k_1=2}^{d}\sum_{k_2\neq k_1}^{d} E(Y_{n,k_1}^2 Y_{n,k_2}^2) - 1.
\end{align*}

\begin{comment}
    \textcolor{blue}{Besides, when $k_1 \ne k_2$, it is easy to verify that $E(U_{k_1 l_1} U_{k_1 l_2} U_{k_2 l_3} U_{k_2 l_4}) \ne 0$
only if $l_1=l_2$ and $l_3=l_4$. Indeed, if $l_1\ne l_2$ or
$l_3\ne l_4$, independence of the corresponding components under $H_0$ allows the expectation to be factorized into products containing terms with zero expectation. Hence, $E(U_{k_1 l_1} U_{k_1 l_2} U_{k_2 l_3} U_{k_2 l_4})=0$. Consequently, only the case $l_1=l_2$ and $l_3=l_4$ contributes, and
therefore
\begin{align*}
    E(Y_{n,k_1}^2Y_{n,k_2}^2)
    &= \operatorname{Var}(\widehat{T}_{nd})^{-2} 
    \sum_{l_1=1}^{k_1-1} \sum_{l_2=1}^{k_2-1}
    E(U_{k_1 l_1}^2U_{k_2 l_2}^2) \\
    &= \operatorname{Var}(\widehat{T}_{nd})^{-2}
    \sum_{l_1=1}^{k_1-1} \sum_{l_2=1}^{k_2-1}
    E(U_{k_1 l_1}^2) E(U_{k_2 l_2}^2).
\end{align*}}
\end{comment}

By the independence of distinct columns and the
conditional centering of $U_{kl}$, a product has zero
expectation whenever a column appears in only one
of its four factors. Consequently, only terms corresponding to two
repeated pairs, four distinct pairs forming a cycle,
or a repeated pair together with the other two
edges of a triangle need to be considered.
Collecting these terms yields

\begin{align*}
    E(Y_{n,k_1}^2Y_{n,k_2}^2)
    &= 16\operatorname{Var}(\widehat{T}_{nd})^{-2} 
    \sum_{l_1=1}^{k_1-1} \sum_{l_2=1}^{k_1-1}\sum_{l_3=1}^{k_2-1} \sum_{l_4=1}^{k_2-1}
    E(U_{k_1l_1}U_{k_1l_2}U_{k_2l_3}U_{k_2l_4}) %\\&= 16\operatorname{Var}(\widehat{T}_{nd})^{-2} \Bigg( \sum_{1\leq l_1, l_2, l_3, l_4 \leq \min \{k_1 - 1, k_2 - 1\}}E(U_{k_1l_1}U_{k_1l_2}U_{k_2l_3}U_{k_2l_4}) \\ &+ \sum_{\substack{1\leq l_1, l_2 \leq \min \{k_1 - 1, k_2 - 1\} \\ \min \{k_1 - 1, k_2 - 1\} < l_3, l_4 \leq \max \{k_1-1, k_2-1\} }} E(U_{k_1l_1}U_{k_1l_2}U_{k_2l_3}U_{k_2l_4})\Bigg) \\ &= 16\operatorname{Var}(\widehat{T}_{nd})^{-2} \sum_{1\leq l_1, l_2, l_3, l_4 \leq \min \{k_1 - 1, k_2 - 1\}}E(U_{k_1l_1}U_{k_1l_2}U_{k_2l_3}U_{k_2l_4}) 
    \\ &= 16\operatorname{Var}(\widehat{T}_{nd})^{-2}  \Bigg (\sum_{ l_1= 1}^{k_1- 1}\sum_{l_2=1}^{k_2-1}
    E(U_{k_1 l_1}^2U_{k_2 l_2}^2) + 2\sum_{1\leq l_1 \neq l_2 \leq\min \{k_1-1, k_2-1\}} E(U_{k_1l_1}U_{k_1l_2}U_{k_2l_1}U_{k_2l_2})  \\ &+ 2\sum_{l_1=1}^{\min\{k_1-1, k_2-1\}} E(U_{\min\{k_1, k_2\}l_1}^2U_{\max\{k_1, k_2\}l_1} U_{\max\{k_1, k_2\} \min\{k_1, k_2\}})\Bigg)
\end{align*}

Let $\mathcal{J}_{k_1, k_2} = \{l_1, l_2 \mid 1\leq l_1 \leq k_1 -1, \, 1\leq l_2 \leq k_2 -1\} $ and $\mathcal{J}_{k_1, k_2}^{'} = \{l_1, l_2 \mid 1\leq l_1 \leq k_1 -1, \, 1\leq l_2 \leq k_2 -1, \, l_1 \neq k_2, \, l_2\neq k_1 , \, l_1 \neq l_2 \}$.

We obtain, 
\begin{equation}\label{proizvod_ocekivanja}
    \sum_{ l_1= 1}^{k_1- 1}\sum_{l_2=1}^{k_2-1}
    E(U_{k_1 l_1}^2U_{k_2 l_2}^2) = \sum_{(l_1, l_2) \in \mathcal{J}_{k_1, k_2}^{'}} E(U_{k_1 l_1}^2)E(U_{k_2 l_2}^2) + \sum_{(l_1, l_2) \in \mathcal{J}_{k_1, k_2} \setminus \mathcal{J}_{k_1, k_2}^{'}} E(U_{k_1 l_1}^2U_{k_2 l_2}^2).
\end{equation}

We next establish the bound of $E(U_{k_1l_1}U_{k_1l_2}U_{k_2l_1}U_{k_2l_2})$. Expanding the product and using the independence
of distinct columns, we obtain
\[
\begin{aligned}
&E(U_{k_1l_1}U_{k_1l_2}U_{k_2l_1}U_{k_2l_2})
=\frac{1}{n^4(n-1)^4(n-2)^4(n-3)^4} \nonumber\\&\times \sum_{\substack{1\leq i_1 \neq i_2 \neq i_3 \neq i_4 \leq n \\ 1\leq i_1^{'} \neq i_2^{'} \neq i_3^{'} \neq i_4^{'} \leq n \\ 1\leq i_1^{''} \neq i_2^{''} \neq i_3^{''} \neq i_4^{''} \leq n \\ 1\leq i_1^{'''} \neq i_2^{'''} \neq i_3^{'''} \neq i_4^{'''} \leq n }} 
        \begin{aligned}[t]
        &E[f_{k_1}(i_1, i_2)f_{k_1}(i_3, i_4)f_{k_1}(i_1^{'}, i_2^{'})f_{k_1}(i_3^{'}, i_4^{'})] \\ &E[f_{l_2}(i_1^{'}, i_2^{'})f_{l_2}(i_3^{'}, i_4^{'}) f_{l_2}(i_1^{'''}, i_2^{'''})f_{l_2}(i_3^{'''}, i_4^{'''})] \\ &E[f_{k_2}(i_1^{''}, i_2^{''})f_{k_2}(i_3^{''}, i_4^{''})f_{k_2}(i_1^{'''}, i_2^{'''})f_{k_2}(i_3^{'''}, i_4^{'''})] \\ &
        E[f_{l_1}(i_1, i_2)f_{l_1}(i_3, i_4)f_{l_1}(i_1^{''}, i_2^{''})f_{l_1}(i_3^{''}, i_4^{''})].
        \end{aligned}
\end{aligned}
\]

Let \(I, I', I'', I'''\) denote the sets of four observation indices appearing in \((i_1,i_2,i_3,i_4)\), \((i'_1,i'_2,i'_3,i'_4)\), \((i''_1,i''_2,i''_3,i''_4)\), and \((i'''_1,i'''_2,i'''_3,i'''_4)\), respectively. By the first-order degeneracy of the kernel, each of the four expectations in the above expansion vanishes whenever the corresponding two index sets have at most one element in common. Consequently, for a summand to be nonzero, it is necessary that $|I\cap I'|\geq 2,\,
|I\cap I''|\geq 2, \, |I'\cap I'''|\geq 2$ and $|I''\cap I'''|\geq 2.$ Since the first index vector contains four distinct indices, it follows that the second index vector can introduce at most two new indices, as \(|I\cap I'|\geq 2\). Similarly, the third index vector can introduce at most two new indices because \(|I\cap I''|\geq 2\). Finally, the fourth index vector can introduce at most two new indices, since it must share at least two indices with each of \(I'\) and \(I''\). Hence, every nonzero summand involves at most $10$ distinct observation indices. It follows that the number of potentially nonzero summands is $O(n^{10})$. Since the kernels are uniformly bounded,
each summand is bounded in absolute value by a constant $C$ independent of $n$.
Therefore,
\begin{align}\label{kombinovaniIndeksi}
    \begin{aligned}
\left|
E(U_{k_1l_1}U_{k_1l_2}U_{k_2l_1}U_{k_2l_2})
\right|
&\le
\frac{C_2n^{10}}{n^4(n-1)^4(n-2)^4(n-3)^4}
\le
\frac{C_2}{n^6}.
\end{aligned}
\end{align}
Thus, $\left|
E(U_{k_1l_1}U_{k_1l_2}U_{k_2l_1}U_{k_2l_2})
\right|
=O(n^{-6})$.

For the triangular products, Hölder's inequality yields
\begin{align}\label{eq:trojke}
\Big|E(U_{\min\{k_1, k_2\}l_1}^2 &U_{\max\{k_1, k_2\}l_1} U_{\max\{k_1, k_2\} \min\{k_1, k_2\}})\Big|
\nonumber \\ &\le
\left(EU_{\min\{k_1, k_2\}l_1}^{4}\right)^{1/2}
\left(EU_{\max\{k_1, k_2\}l_1}^{4}\right)^{1/4}
\left(EU_{\max\{k_1, k_2\} \min\{k_1, k_2\}}^{4}\right)^{1/4}
\le \frac{C_3}{n^{4}}.
\end{align}

It follows from \eqref{eq: proizvodKvadrata}, \eqref{kombinovaniIndeksi} and \eqref{eq:trojke}
\begin{align}
    E(Y_{n,k_1}^{2}Y_{n,k_2}^{2}) 
\le 
\frac{16}{\operatorname{Var}(\widehat{T}_{nd})^{2}}
&\Big(\sum_{ l_1= 1}^{k_1- 1}\sum_{l_2=1}^{k_2-1}
    E(U_{k_1 l_1}^2U_{k_2 l_2}^2)  +
\frac{(\min\{k_1, k_2\}-1)C_1}{n^{4}}
\\ &+
\frac{(\min \{k_1, k_2\}-1)(\min \{k_1, k_2\}-2)C_2}{n^{6}}\nonumber + \frac{(\min \{k_1, k_2\}-1)C_3}{n^{4}}
\Big)
\end{align}
and
\begin{align*}
    \sum_{k_1=2}^{d}\sum_{k_2\neq k_1}^{d} E(Y_{n,k_1}^2 Y_{n,k_2}^2) \leq 1 + \frac{16Cn^4}{d^2(d-1)^2} \left(\frac{d^3}{n^4} + \frac{d^4}{n^6} +\frac{d^3}{n^4}\right).%\approx 1 + \frac{1}{n^2} + \frac{2}{d}.
\end{align*}

Therefore, $\sum_{k_1=2}^{d}\sum_{k_2\neq k_1}^{d} E(Y_{n,k_1}^2 Y_{n,k_2}^2)=1 + O\left(\frac{1}{n^2} + \frac{1}{d}\right)$. Finally, 
$E\left(\sum_{k=2}^{d} Y_{n,k}^2 - 1\right)^2  = O\left(\frac{1}{n^2} + \frac{1}{d}\right)$, and  $\sum_{k=2}^dY^2_{n,k}\overset{P}{\to} 1$ as $d, n \to \infty$.

\begin{comment}
    Besides, when $ k_1 \ne k_2 $, it is easy to verify that
$E(U_{k_1 l_1} U_{k_1 l_2} U_{k_2 l_3} U_{k_2 l_4}) \ne 0$ only if $ l_1 = l_2$ and $ l_3 = l_4 $. Therefore,

\begin{align*}
    E(Y_{n,k_1}^2Y_{n,k_2}^2) %&= \operatorname{Var}(\widehat{T}_{nd})^{-2} \sum_{l_1=1}^{k_1 - 1} \sum_{l_2=1}^{k_1 - 1} \sum_{l_3=1}^{k_2 - 1} \sum_{l_4=1}^{k_2 - 1} E(U_{k_1 l_1} U_{k_1 l_2} U_{k_2 l_3} U_{k_2 l_4}) \\ 
    &= \operatorname{Var}(\widehat{T}_{nd})^{-2} \sum_{l_1=1}^{k_1 - 1} \sum_{l_2=1}^{k_2 - 1} E(U_{k_1 l_1}^2U_{k_2 l_2}^2)
    \\ &= \operatorname{Var}(\widehat{T}_{nd})^{-2} \sum_{l_1=1}^{k_1 - 1} \sum_{l_2=1}^{k_2 - 1} E(U_{k_1 l_1}^2) E(U_{k_2 l_2}^2). 
\end{align*}

Hence,

\begin{align*}
    \sum_{k_1=2}^d \sum_{\substack{k_2=2 \\ k_2\neq k_1}}^d EY_{n,k_1}^2Y_{n,k_2}^2 &\!=\! 16\operatorname{Var}(\widehat{T}_{nd})^{-2} \sum_{k_1=2}^d \sum_{l_1=1}^{k_1 - 1} E(U_{k_1,{l_1}}^2) \sum_{k_2=2}^d\sum_{l_2=1}^{k_2 - 1}  E(U_{k_2,l_2}^2) \!=\! 1.  %16\cdot \frac{1}{16} = 1.
\end{align*}

Therefore, by the result in (i), we have
\begin{align*}
    E\left( \sum_{k=2}^{d} Y_{n,k}^2 - 1 \right)^2 \to 0 \quad \text{as } d \to \infty,
\end{align*}
which holds regardless of whether $n$ is fixed or diverging.

\end{comment}
\end{proof}

Now we have all ingredients to prove Theorem \ref{thm}.

\begin{comment}
    \begin{proof}[Proof of Theorem \ref{thm}]
According to Lemma \ref{lemma_martingale}, $\widehat{T}_{nd}$ can be expressed as a sum of a m.d. array or sequence. Therefore, the Central Limit Theorem for m.d. arrays stated in Theorem 2.3 of \cite{mcleish1974dependent} is applicable. 
In fact, condition (i) in Lemma \ref{lemma_cond} corresponds to the Lyapunov condition, which ensures that assumptions (a) and (b) in McLeish’s theorem from \cite{mcleish1974dependent} hold. Moreover, condition (ii) in Lemma \ref{lemma_cond} is equivalent to assumption (c) in McLeish’s theorem. As a result, $\frac{\widehat{T}_{nd}}{\sqrt{\operatorname{Var}(\widehat{T}_{nd})}} \to N(0,1)$ as $d\to \infty$.
\end{proof}

\end{comment}

\begin{proof}[Proof of Theorem \ref{thm}]
According to Lemma \ref{lemma_martingale}, $\widehat{T}_{nd}$ can be expressed as a sum of a m.d. array or sequence. Therefore, the Central Limit Theorem for m.d. arrays stated in Theorem 2.3 of \cite{mcleish1974dependent} is applicable. 
In fact, condition (i) in Lemma \ref{lemma_cond} corresponds to the Lyapunov condition, which ensures that assumptions (a) and (b) in McLeish’s theorem from \cite{mcleish1974dependent} hold. To see this, note that
\[
E\Big[\max_{2\le k\le d} Y_{n,k}^4\Big]\;\le\; E\Big[\sum_{k=2}^d Y_{n,k}^4\Big]\;=\;\sum_{k=2}^d EY_{n,k}^4 \;\to\; 0 ,
\]
which yields both (a) and (b):
\begin{itemize}
    \item[(a)] By Jensen's inequality,
    \[
    E\Big[\max_{k} Y_{n,k}^2\Big] \le \Big(E\Big[\max_{k} Y_{n,k}^4\Big]\Big)^{1/2} \to 0,
    \]
    so $\max_{k}|Y_{n,k}|\to 0$ in $L_2$. A convergent sequence is bounded, hence $\max_{k}|Y_{n,k}|$ is uniformly bounded in $L_2$ norm (at least for all sufficiently large $d$).
    \item[(b)] By Markov's inequality, for every $\varepsilon>0$,
    \[
    P\Big(\max_{k}|Y_{n,k}|>\varepsilon\Big)\le \frac{1}{\varepsilon^4}\,E\Big[\max_{k} Y_{n,k}^4\Big]\le \frac{1}{\varepsilon^4}\sum_{k=2}^d EY_{n,k}^4\to 0 ,
    \]
    that is, $\max_{k}|Y_{n,k}|\overset{P}{\to}0$.
\end{itemize}
Moreover, condition (ii) in Lemma \ref{lemma_cond} is equivalent to assumption (c) in McLeish’s theorem. As a result, $\frac{\widehat{T}_{nd}}{\sqrt{\operatorname{Var}(\widehat{T}_{nd})}} \to N(0,1)$ as $d\to \infty$.
\end{proof}

%\bibliographystyle{abbrv}
%\bibliography{references}

\end{document}